\documentclass[aps,pra,twocolumn,superscriptaddress,floatfix,
nofootinbib,showpacs,longbibliography]{revtex4-2}

\usepackage[normalem]{ulem}
\usepackage{xcolor}
\usepackage[utf8]{inputenc}  
\usepackage[T1]{fontenc}     
\usepackage[british]{babel}  
\usepackage[sc,osf]{mathpazo}
\usepackage[scaled=0.86]{berasans}  
\usepackage[colorlinks=true, citecolor=blue, urlcolor=blue]{hyperref}  
\usepackage{graphicx} 
\usepackage[babel]{microtype}  
\usepackage{amsmath,amssymb,amsthm,bm,amsfonts,mathrsfs,bbm} 

\usepackage{xspace}  
\definecolor{wrwrwr}{rgb}{0.3803921568627451,0.3803921568627451,0.3803921568627451}
\definecolor{rvwvcq}{rgb}{0.08235294117647059,0.396078431372549,0.7529411764705882}
\usepackage{pgf,tikz,pgfplots}
\usetikzlibrary{calc}
\pgfplotsset{compat=1.15}
\usepackage{mathrsfs}
\usetikzlibrary{arrows,snakes}
\usetikzlibrary{positioning}
\usepackage{xcolor}
\usepackage{appendix}
\usepackage{multirow}
\usepackage{array}
\usepackage{bigstrut}
\usepackage{braket}
\usepackage{color}
\usepackage{natbib}
\usepackage{multirow}
\usepackage{float}
\usepackage[caption = false]{subfig}
\usepackage{xcolor,colortbl}
\usepackage{color}

\newcommand{\be}{\begin{equation}}
	\newcommand{\ee}{\end{equation}}
\newcommand{\ba}{\begin{eqnarray}}
	\newcommand{\ea}{\end{eqnarray}}

\newcommand{\topt}[2]{\ket{#1}\ket{#2}}

\newtheorem{theorem}{Theorem}
\newtheorem{corollary}{Corollary}
\newtheorem{definition}{Definition}
\newtheorem{proposition}{Proposition}
\newtheorem{observation}{Observation}

\definecolor{rvwvcq}{rgb}{0,0,1}

\def\>{\rangle}
\def\<{\langle}

\newcommand{\tri}[3]{\ket{#1}\ket{#2}\ket{#3}}

\makeatletter
\DeclareRobustCommand{\cev}[1]{%
	{\mathpalette\do@cev{#1}}%
}
\newcommand{\do@cev}[2]{%
	\vbox{\offinterlineskip
		\sbox\z@{$\m@th#1 x$}%
		\ialign{##\cr
			\hidewidth\reflectbox{$\m@th#1\vec{}\mkern4mu$}\hidewidth\cr
			\noalign{\kern-\ht\z@}
			$\m@th#1#2$\cr
		}%
	}%
}
\makeatother

\begin{document}
	



\title{Multiparty orthogonal product states with minimal genuine nonlocality}
	
\author{Sumit Rout}
\affiliation{International Centre for Theory of Quantum Technologies (ICTQT), University of Gdańsk, 80-308 Gdańsk, Poland}
	
\author{Ananda G. Maity}
\affiliation{S.N. Bose National Center for Basic Sciences, Block JD, Sector III, Salt Lake, Kolkata 700106, India.}

\author{Amit Mukherjee}
\affiliation{S.N. Bose National Center for Basic Sciences, Block JD, Sector III, Salt Lake, Kolkata 700106, India.}
		
\author{Saronath Halder}
\affiliation{Harish-Chandra Research Institute, HBNI, Chhatnag Road, Jhunsi, Allahabad 211019, India.} 
	
\author{Manik Banik}
\affiliation{School of Physics, IISER Thiruvananthapuram, Vithura, Kerala 695551, India.}

\begin{abstract}
Nonlocality without entanglement and its subsequent generalizations offer deep information-theoretic insights and subsequently find several useful applications. Concept of genuinely nonlocal set of product states emerges as a natural multipartite generalization of this phenomenon. Existence of such sets eventually motivates the problem concerning their entanglement-assisted discrimination. Here, we construct examples of genuinely nonlocal product states for  arbitrary number of parties. Strength of genuine nonlocality of these sets can be considered minimal as their perfect discrimination is possible with entangled resources residing in Hilbert spaces having the smallest possible dimensions. Our constructions lead to fully separable measurements that are impossible to implement even if all but one party come together. Furthermore, they also provide the opportunity to compare different multipartite states that otherwise are  incomparable under single copy local manipulation.
\end{abstract}



\maketitle
\section{Introduction} 
Quantum entanglement has been established as the useful resource for numerous practical tasks, starting form advanced means of communication \cite{Bennett92,Bennett93,Chiribella21}, improved metrology and estimation \cite{Walther04,Mitchell04,Joo11,Giovannetti11} to randomness processing \cite{Pironio10,Colbeck12,Liu21}. For multipartite systems, entanglement appears in different inequivalent forms \cite{Greenberger90,Dur00,Verstraete02} and accordingly finds more exotic applications \cite{Dur99,Giovannetti04,Helwig12,Bhattacharya21}. Characterization, quantification, and detection of quantum entanglement therefore have practical relevance and it vastly shaped the research direction in quantum information theory during last
three decades (see \cite{Guhne09,Horodecki09} and references therein). {Entanglement also lies at the core of almost all foundational debates in quantum theory \cite{Einstein35,Bohr35,Schrodinger35,Bell64,Bell66,Wiseman07,Pusey12,Bong20}. In particular, it is crucial to establish the puzzling nonlocal feature of quantum theory. J. S. Bell, in his seminal result \cite{Bell64,Bell66}, derived an experimentally testable criterion that any {\it local-realistic} theory must satisfy, whereas quantum statistics obtained from suitably chosen local measurements performed on properly chosen entangled state can violate this inequality and hence establish nonlocal feature of quantum theory. Several experiments with variety of quantum systems have reported positive Bell test and thus ensure nonlocal nature of quantum world \cite{Aspect82,Hensen15,BIG18,Rauch18}.}

Entanglement has also been proved to be advantageous in hypothesis testing and discrimination tasks \cite{Piani09,Hirche21,Pirandola19,Takagi19}. A particular interest is the local state discrimination problem, where the aim is to identify a multipartite quantum state, drawn randomly from a known set of states, under the operational paradigm of local operation and classical communication (LOCC). {In such a scenario, quantum theory exhibits a different kind of nonlocal behaviour that involves no entanglement and is distinct from Bell nonlocality. In a seminal paper, Bennett {\it et al.} provide examples of orthogonal product bases for multipartite system \cite{Bennett99} that are locally indistinguishable. They coined the term `quantum nonlocality without entanglement' for this phenomenon as perfect discrimination of the states and requires `nonlocal' (read as global/joint) measurement on the composite system. Subsequently this result motivates plethora of research in general local state discrimination problem \cite{Walgate00,Ghosh01,Walgate02,Horodecki03,Watrous05,Hayashi06,Niset06,Duan07,Calsamiglia10,Bandyopadhyay11,Halder18,Halder19(1),Agrawal19,Bhattacharya20,Banik21} and in this work our study will also deal with this particular kind of nonlocal behaviour of quantum theory.} A locally indistinguishable mutually orthogonal set of states can be distinguished perfectly if entangled states are provided as resource along with LOCC. For instance, Bennett {\it et al.}'s nonlocal product basis of $(\mathbb{C}^3)^{\otimes2}$ system can be perfectly distinguished if a maximally entangled state in this Hilbert space is provided as resource. The seminal teleportation protocol \cite{Bennett93} makes the discrimination task viable. Quite surprisingly, in subsequent work, Cohen showed that a two-qutrit maximally entangled state is not necessary for perfect discrimination of this nonlocal product basis; instead, a two-qubit maximally entangled state suffices the purpose \cite{Cohen08}. Cohen's protocol offers an efficient use of the costly entangled resource in local state discrimination problem.

Recently, a stronger notion of nonlocality without entanglement phenomena is identified for multipartite quantum systems \cite{Halder19} which subsequently motivates renewed interest in constructing nonlocal product set of states for multipartite systems as well as their entanglement assisted discrimination \cite{Rout19,Zhang19,Jiang20,Shi20,Yuan20}. In this letter, we first present a set of tripartite product states which is locally indistinguishable given arbitrary amount of entanglement shared between any two of the three parties. In other words, the set remains indistinguishable even if any two of the parties come together but do not share any entanglement with the third party. Therefore, the set requires genuinely multipartite entangled resource for perfect discrimination when all the parties are spatially separated. Interestingly, we show that given a three qubit GHZ state as resource, the states can be perfectly distinguished although they live in $\mathbb{C}^4\otimes\mathbb{C}^3\otimes\mathbb{C}^3$ dimensional system. Note that the resource used here is much cheaper than teleportation based resource ($2$ copies of two-qutrit maximally entangled state in this case). In fact our protocol uses the minimal dimensional genuinely entangled resource and hence the nonlocal strength of the constructed set of states can be considered minimal. We then generalize the construction for arbitrary number of spatially separated parties and also discuss its entanglement assisted discrimination. For the $n$ partite case, the construction lives in $\mathbb{C}^{n+1}\otimes(\mathbb{C}^3)^{\otimes n-1}$. Moreover, an $n$-qubit GHZ state suffices as resource for their perfect discrimination which again turns out to be the minimal dimensional resource. Our construction also provides an operational way to compare different classes of multipartite entanglement that otherwise are incomparable under LOCC.

\section{Preliminaries} 
Although the history of quantum state discrimination dates back to early 1970's with an initial attempt to formulate information protocols using quantum optical devices \cite{Helstrom69,Holevo73,Yuen75}, local state discrimination (LSD) problem gained research interest much later \cite{Peres91,Massar95,Bennett99}. Given only one copy of the system, it asks to identify the state chosen randomly from a known ensemble of states $\{p_i,\ket{\psi_i}\}_{i=1}^m$ under the restriction that the spatially separated parties can perform only LOCC; where $\forall~i~,\ket{\psi_i}\in\otimes_{j=1}^n\mathcal{H}_{j}$ with $\mathcal{H}_{j}$ being the Hilbert space of the $j^{th}$ subsystem. In a product LSD problem, all $\ket{\psi_i}$'s are considered to be fully product states, \textit{i.e.}, $\forall~ i,~\ket{\psi_i}=\otimes_{j=1}^n\ket{\phi^j_i}$ with $\ket{\phi^j_i}\in\mathcal{H}_j$. 
\begin{definition}\label{def1}
Nonlocal product states (NPS): A set of mutually orthogonal and fully product states $\mathbb{S}:=\left\{\ket{\psi_i}\right\}_{i=1}^K\subset\otimes_{j=1}^n\mathcal{H}_{j}$ will be referred to as NPS if they cannot be perfectly distinguished under LOCC when all the parties are spatially separated. 
\end{definition}
\begin{definition}\label{def2}
Genuinely nonlocal product states (GNPS): A set of mutually orthogonal and fully product states $\mathbb{S}:=\left\{\ket{\psi_i}\right\}_{i=1}^K\subset\otimes_{j=1}^n\mathcal{H}_{j}$ will be referred to as GNPS if they cannot be locally distinguished in any possible bipartition. \end{definition}
Note that the above definition captures the strongest possible notion of nonlocality without entanglement phenomenon for multipartite systems. The states of a GNPS can neither be locally distinguished in any `$n-1$ {\it vs} $1$' bipartition nor they can be locally distinguished in any `$n-k$ {\it vs} $k$' bipartition, with arbitrary $k$ parties grouping together. Clearly every GNPS is a NPS, but the converse is not true in general. For instance, the SHIFTS UPB of $(\mathbb{C}^2)^{\otimes3}$ as constructed in \cite{Bennett99upb,Bennett99} is a tripartite NPS but not a GNPS. In this work, our primary aim is to construct GNPS and then study their entanglement assisted discrimination. Before discussing our construction, we first recall an example of bipartite NPS which is given by,    
\begin{align*}
\mathbb{S}_{Ben}\equiv\left\lbrace \topt{0}{\eta_\pm},\topt{\eta_\pm}{2},\topt{2}{\xi_\pm},\topt{\xi_\pm}{0}\right\rbrace\subset(\mathbb{C}^3)^{\otimes2},
\end{align*}
where $\ket{\eta_\pm}:=\left(\ket{0}\pm\ket{1}\right)/\sqrt{2}$ and $\ket{\xi_\pm}:=\left(\ket{1}\pm\ket{2}\right)/\sqrt{2}$. As pointed out in \cite{Bennett99}, deletion of any state form $\mathbb{S}_{Ben}$ makes the remaining set locally distinguishable, whereas if we add another orthogonal product state to it, for instance the state $\ket{1}\ket{1}$, the resulting set remains nonlocal. This fact can be further generalized. For this purpose, first note that, two sets of states $\mathbb{S}$ and $\mathbb{S}^\prime$ are called orthogonal if and only if $\braket{\phi|\phi^\prime}=0,~\forall~\ket{\phi}\in\mathbb{S},~{and}~\ket{\phi^\prime}\in\mathbb{S}^\prime$; and they will be denoted as $\mathbb{S}\perp\mathbb{S}^\prime$.
\begin{observation}\label{obs1}
Let $\mathbb{S}\subset\otimes_{j=1}^n\mathcal{H}_j$ be a multipartite NPS / GNPS. The set of states $\mathbb{A}:=\mathbb{S}\cup\mathbb{S}^\prime$ is a NPS / GNPS for any set of mutually orthogonal states $\mathbb{S}^\prime$ such that $\mathbb{S}\perp\mathbb{S}^\prime$.
\end{observation}
Proof of this observation trivially follows an argument of \textit{reductio ad absurdum}. 
If $\mathbb{A}$ were a locally distinguishable set then for every $\ket{\psi} \in \mathbb{A} $ chosen at random, it is possible to perfectly identify this state under LOCC. This should hold even when the state lies in the nonlocal set $\mathbb{S}$ which leads to a contradiction.

Given a set of states $\chi:=\{\ket{\beta}_i~|~i=1,\cdots K\}$ and another state $\ket{\alpha}$ let us define $\chi\otimes\ket{\alpha}:=\{\ket{\beta}_i\otimes\ket{\alpha}~|~i=1,\cdots K\}$. With this notation, we will now put our next observation which will be relevant in subsequent proofs.
\begin{observation}\label{obs2}
Let $\mathbb{S}\subset\otimes_{j=1}^n\mathcal{H}_j$ be a multipartite NPS / GNPS. Consider the set $\mathbb{S'}:=\mathbb{S}\otimes\ket{\phi_0}_{a_1 \cdots a_m}$, where $\ket{\phi_0}_{a_1 \cdots a_m}$ is some fully separable state with some of the subsystems $\{a_i\}$ is in possession with $i^{th}$ party. The resulting set $\mathbb{S'}$ is again an NPS / GNPS with respect to the same multipartite configuration.   
\end{observation}
Observation \ref{obs2} follows from the fact that any fully separable state can always be prepared locally \cite{Rinaldis04}.

\section{Results} 
With the aforesaid observations in hand, in the following, we first construct a tripartite GNPS. 
\begin{figure}[t!]
\begin{center}
\includegraphics[scale=0.32]{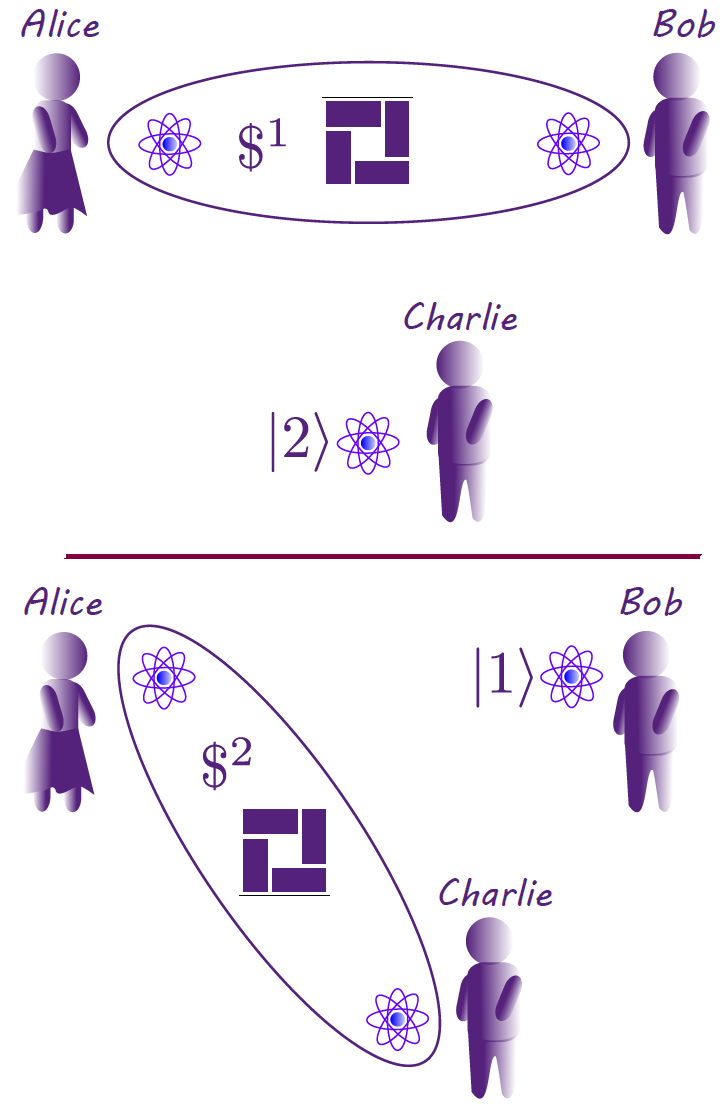}
\caption{Tripartite GNPS $\mathbb{G}[4\otimes3^{\otimes2}]$. Alice \& Bob share the set of states $\$^1\equiv\left\lbrace \ket{\zeta_{\pm}^0},\ket{\zeta^{1,1}_{\pm}},\ket{\zeta^{1,2}_{\pm}},\ket{\zeta^{1,3}_{\pm}}\right\rbrace$ while Charlie's state be $\ket{2}$. Alice \& Charlie share the set of states $\$^2\equiv\left\lbrace \ket{\zeta_{\pm}^0},\ket{\zeta^{2,1}_{\pm}},\ket{\zeta^{2,2}_{\pm}},\ket{\zeta^{2,3}_{\pm}}\right\rbrace$ while Bob's state be $\ket{1}$. Clearly, $\mathbb{G}[4\otimes3^{\otimes2}]\equiv\left\lbrace\$^1_{AB}\otimes\ket{2}_C\right\rbrace\cup\left\lbrace\$^2_{AC}\otimes\ket{1}_B\right\rbrace$.}\label{fig1}
\end{center}
\end{figure}
\begin{proposition}\label{prop1}
The set of states $\mathbb{G}[4\otimes3^{\otimes2}]$ defined below is a GNPS in $\mathbb{C}^4_A\otimes\mathbb{C}^3_B\otimes\mathbb{C}^3_C$.
\footnotesize
\begin{align}\nonumber
\rotatebox[origin=c]{0}{$\mathbb{G}[4\otimes3^{\otimes2}]\equiv$}
\left\{\!\begin{aligned}
\ket{\zeta_{\pm}^0}~&:=\tri{\epsilon_\pm}{1}{2},~~
\ket{\zeta^{1,1}_{\pm}}:=\tri{1}{\gamma_\pm^1}{2},\\
\ket{\zeta^{1,2}_{\pm}}&:=\tri{\gamma_\pm^1}{p}{2},~
\ket{\zeta^{1,3}_{\pm}}:=\tri{p}{\epsilon_\pm}{2},\\
\ket{\zeta^{2,1}_{\pm}}&:=\tri{2}{1}{\gamma_\pm^2},~
\ket{\zeta^{2,2}_{\pm}}:=\tri{\gamma_\pm^2}{1}{p},\\
&~~~~~~~\ket{\zeta^{2,3}_{\pm}}:=\tri{p}{1}{\epsilon_\pm}~
\end{aligned}\right\}.	
\end{align}
\normalsize
Here $\{\ket{p},\ket{q},\ket{1},\ket{2}\}$ are mutually orthogonal states, $\ket{\epsilon_\pm}:=\frac{1}{\sqrt{2}}(\ket{p}\pm\ket{q})$ \& $\ket{\gamma_\pm^i}:=\frac{1}{\sqrt{2}}(\ket{q}\pm\ket{i})$; $i=1,2$. 
\end{proposition} 
\begin{proof}
Consider the subset of states $\$^i\equiv\left\lbrace \ket{\zeta_{\pm}^0},\ket{\zeta^{i,1}_{\pm}},\ket{\zeta^{i,2}_{\pm}},\ket{\zeta^{i,3}_{\pm}}\right\rbrace\subset\mathbb{G}[4\otimes3^{\otimes2}]$, for $i\in\{1,2\}$. The set $\$^1$ has {analogous structure as the set $\mathbb{S}_{Ben}$ (with a notation change, $p\rightarrow 0, q\rightarrow 1, 1\rightarrow 2$)} between Alice and Bob while Charlie has the fixed state $\ket{2}$ (see Fig.\ref{fig1}). This, along with the Observations \ref{obs1} \& \ref{obs2}, assure that the set $\mathbb{G}[4\otimes3^{\otimes2}]$ cannot be locally discriminated even when Charlie groups with either Alice or Bob. Similarly the set $\$^2$ prohibits perfect local discrimination of $\mathbb{G}[4\otimes3^{\otimes2}]$ even when Alice and Bob are grouped together. This completes the proof. \end{proof}

At this point, a pertinent question is how to quantify the amount of `genuine nonlocality without entanglement' for a given GNPS? Note that, given sufficient amount of entanglement among the spatially separated parties any GNPS can be perfectly distinguished. For instance discrimination of $\mathbb{G}[4\otimes3^{\otimes2}]$ is possible given $2$ copies of two-qutrit maximally entangled states -- one shared between Alice-Bob and the other between Alice-Charlie. Since entanglement is a costly resource it is therefore relevant to go for an cost efficient discrimination protocol. Given two GNPSs a natural ordering of their strength of `genuine nonlocality without entanglement' can be made from the amount of entanglement required for their perfect discrimination. While the teleportation based discrimination of the $\mathbb{G}[4\otimes3^{\otimes2}]$ required $2$ copies of two-qutrit maximally entangled states that lives in the Hilbert space $\mathbb{C}^9_A\otimes\mathbb{C}^3_B\otimes\mathbb{C}^3_C$, we will now discuss a much cost efficient discrimination protocol. In particular we will show that an entanglement resource living in the Hilbert space $(\mathbb{C}^2)^{\otimes3}$ will suffices for perfect discrimination of the set. 
\begin{theorem}\label{theo1}
The set of states $\mathbb{G}[4\otimes3^{\otimes2}]$ can be perfectly discriminated locally when the state $\ket{g_3}:=(\ket{000}+\ket{111})/\sqrt{2}$ is shared as resource.  
\end{theorem}
\begin{proof}
We will associate block letter party index with the states that need to be distinguished and denote the resource state as $\ket{g_3}_{abc}=(\ket{000}_{abc}+\ket{111}_{abc})/\sqrt{2}$. Local distinguishability of the set $\mathbb{G}[4\otimes3^{\otimes2}]$ boils down to identify the pairs $\left\lbrace \ket{\zeta_\pm}\right\rbrace $ preserving the post-measurement orthogonality between $\ket{\zeta_+}$ and $\ket{\zeta_-}$, as the result in \cite{Walgate00} assures local distinguishability between any two orthogonal states. The discrimination protocol proceeds as follows.\\
	{\bf Step-1:} Alice performs the measurement $\mathcal{M}\equiv\{M,\mathbb{I}-M\}$, where $M:=\mathbb{P}\left[\ket{p}_A;\ket{0}_{a}\right]+\mathbb{P}\left[\left( \ket{q},\ket{1},\ket{2}\right)_A;\ket{1}_{a}\right]$. Here, we use the notation $\mathbb{P}\left[\left(\ket{e},\ket{f},\cdots\right)_K;\left(\ket{x},\ket{y},\cdots\right)_k\right]:=\left(\ket{e}\bra{e}+\ket{f}\bra{f}+\cdots\right)_K\otimes\left(\ket{x}\bra{x}+\ket{y}\bra{y}+\cdots\right)_k$. Suppose, the projector $M$ clicks. The state $\ket{\zeta}_{ABC}\otimes\ket{g}_{abc}$ evolves to either $\ket{\zeta}_{ABC}\otimes\ket{000}_{abc}$, or $\ket{\zeta}_{ABC}\otimes\ket{111}_{abc}$, or it becomes entangled, where $\ket{\zeta}_{ABC}\in\mathbb{G}[4\otimes3^{\otimes2}]$. Complete list of the evolved states are given below,
	\footnotesize
	\begin{align}\nonumber
	\rotatebox[origin=c]{0}{}
	\left\{\!\begin{aligned}
	\left\lbrace\ket{\zeta^{1,3}_{\pm}},\ket{\zeta^{2,3}_{\pm}} \right\rbrace_{ABC}\otimes\ket{000}_{abc},\\
	\left\lbrace \ket{\zeta^{1,1}_{\pm}},\ket{\zeta^{1,2}_{\pm}},\ket{\zeta^{2,1}_{\pm}},\ket{\zeta^{2,2}_{\pm}}\right\rbrace_{ABC}\otimes\ket{111}_{abc},\\
	\ket{\zeta_{\pm}^0}_{ABC}\Rightarrow\left(\ket{p}_A\ket{000}_{abc}\pm\ket{q}_A\ket{111}_{abc}\right)\ket{1}_{B}\ket{2}_{C}
	\end{aligned}\right\},	
	\end{align}	
	\normalsize	
	{\bf Step-2:} Bob and Charlie respectively perform the measurement,
	\footnotesize
	\begin{align}\nonumber
	\rotatebox[origin=c]{0}{$\mathcal{K}\equiv$}
	\left\{\!\begin{aligned}
	K_1:=\mathbb{I}-K_2-K_3,~K_2:=\mathbb{P}\left[\ket{p}_{B};\ket{1}_{b}\right],~\\
	K_3:=\mathbb{P}\left[\left(\ket{p},\ket{q}\right)_{B};\ket{0}_{b}\right],~~~~~~~~~~~
	\end{aligned}\right\},\\\nonumber
	\rotatebox[origin=c]{0}{$\mathcal{N}\equiv$}
	\left\{\!\begin{aligned}
	N_1:=\mathbb{I}-N_2-N_3,~N_2:=\mathbb{P}\left[\ket{p}_{C};\ket{1}_{c}\right],~\\
	N_3:=\mathbb{P}\left[\left(\ket{p},\ket{q}\right)_{C};\ket{0}_{c}\right],~~~~~~~~~~~
	\end{aligned}\right\}.	
	\end{align}
	\normalsize
	If $K_3$ clicks the state is one of $\left\lbrace \ket{\zeta^{1,3}_{\pm}} \right\rbrace$, if $K_2$ clicks the state is one of $\left\lbrace \ket{\zeta^{1,2}_{\pm}}\right\rbrace $, if $N_3$ clicks the state is one of $\left\lbrace \ket{\zeta^{2,3}_{\pm}} \right\rbrace$, and if $N_2$ clicks the state is one of $\left\lbrace \ket{\zeta^{2,2}_{\pm}}\right\rbrace $. When both $K_1$ and $N_1$ click the state is one of $\left\lbrace \ket{\zeta_{\pm}^0},\ket{\zeta^{1,1}_{\pm}},\ket{\zeta^{2,1}_{\pm}}\right\rbrace $. Obtaining the outcome results from Bob and Charlie, Alice performs the following measurement,
	\footnotesize
	\begin{align}\nonumber
	\rotatebox[origin=c]{0}{$\mathcal{M}^\prime\equiv$}
	\left\{\!\begin{aligned}
	M^\prime_1:=\mathbb{P}\left[\ket{1}_{A};\mathbb{I}_{a}\right],~
	M^\prime_2:=\mathbb{P}\left[\ket{2}_{A};\mathbb{I}_{a}\right],\\
	M^\prime_0:=\mathbb{I}-M^\prime_1-M^\prime_2.~~~~~~~~~~~~
	\end{aligned}\right\}.	
	\end{align}
	\normalsize
	If $M^\prime_1$ clicks the state is one of $\left\lbrace \ket{\zeta^{1,1}_{\pm}}\right\rbrace  $, if $M^\prime_2$ clicks the state is one of $\left\lbrace \ket{\zeta^{2,1}_{\pm}}\right\rbrace $, else it is one of  $\left\lbrace \ket{\zeta_{\pm}^0}\right\rbrace$. If $\mathbb{I}-M$ clicks in {\bf Step-1} then a similar protocol will follow. 
\end{proof}
Theorem \ref{theo1}, for the first time establishes nontrivial and efficient use of the three-qubit GHZ state in product state discrimination problem under LOCC. The GNPS in Proposition \ref{prop1} therefore possesses minimal genuine nonlocality without entanglement as the entangled resource required for its perfect discrimination lives in minimal dimensional Hilbert space. Although the discrimination resource is minimal in the sense of Hilbert space dimension, here a question still remains open whether a state $\alpha\ket{000}+\beta\ket{111}\in(\mathbb{C}^2)^{\otimes3}$ with $\alpha\neq\beta$ suffices perfect discrimination of the set $\mathbb{G}[4\otimes3^{\otimes2}]$ \footnote{Our intuition, in-fact, eventuates from the study of Cohen's work \cite{Cohen08}. Following his technique, local distinguishability of the set $\mathcal{S}_{Ben}$ can be analyzed with the resource state $\alpha\ket{00}+\beta\ket{11}~(\alpha\neq\beta)$, which seems not to provide a perfect success. However, a protocol independent proof of this conviction is not known yet.}. Such a resource is less costlier as it has less 3-tangle \cite{Coffman00} than the state with $\alpha=\beta$. Our intuition is that possibly a state with $\alpha\neq\beta$ will not be sufficient for perfect discrimination of $\mathbb{G}[4\otimes3^{\otimes2}]$. It is extremely difficult to explore all the possible LOCC protocols assisted with such a resource. Therefore, answering this question requires a protocol independent argument which we leave here as an open question for future research.            
  
We now move on to some other consequence of the above construction. Note that, the following set of $22$ orthonormal product states 
\footnotesize
\begin{align}\nonumber
\rotatebox[origin=c]{0}{$\mathbb{G}^C[4\otimes3^{\otimes2}]\equiv$}
\left\{\!\begin{aligned}
\ket{qq2},~\ket{qqq},~\ket{qqp},~\ket{q1q},~\ket{qpq},~\ket{ppp},\\
\ket{1qq},~\ket{2qq},~\ket{pqq},~\ket{1qp},~\ket{2qp},~\ket{pqp},\\
\ket{1pq},~\ket{2pq},~\ket{ppq},~\ket{1pp},~\ket{2pp},~\ket{qpp},\\
\ket{11p},~\ket{11q},~\ket{2q2},~\ket{2p2}~~~~~~~~~~~~
\end{aligned}\right\}	
\end{align}
\normalsize
span the subspace orthogonal to the subspace spanned by $\mathbb{G}[4\otimes3^{\otimes2}]$, and hence the set of states $\mathbb{P}[4\otimes3^{\otimes2}]:=\mathbb{G}[4\otimes3^{\otimes2}]\cup\mathbb{G}^C[4\otimes3^{\otimes2}]$ with adequate normalization constitutes an orthonormal product basis (ONPB) for the Hilbert space $\mathbb{C}^4\otimes(\mathbb{C}^3)^{\otimes2}$; here $\ket{xyz}:=\ket{x}_A\otimes\ket{y}_B\otimes\ket{z}_C$. Manifestly, this ONPB has the property of genuine nonlocality and accordingly it constitutes a fully separable measurement that cannot be implemented even when any two parties come together. It is not hard to argue that the discriminating resource of $\mathbb{G}[4\otimes3^{\otimes2}]$ suffices for discrimination of the set $\mathbb{P}[4\otimes3^{\otimes2}]$. However, at this point, a more difficult question is how much resource is necessary for implementation of the corresponding fully separable measurement. Presently we have no idea regarding the resource requirement and welcome further research in this direction. In rest of the sections, we rather consider multipartite generalization of the above construction. 
\begin{table}[b!]
\centering
\begin{tabular}{ccccc|c}
\hline 
Alice & ~~~~~Bob-1 &~~~~~ Bob-2 &~~~~~ $\cdots$&~~~~~ Bob-m&\\ 
\hline\hline
$\ket{\epsilon_{\pm}}$ &~~~~~$\ket{1}$  &~~~~~$\ket{2}$  &~~~~~$\cdots$&~~~~~ $\ket{m}$ &$:=\ket{\zeta_{\pm}^0}$\\ 
$\ket{1}$ &~~~~~$\ket{\gamma_{\pm}^1}$  &~~~~~$\ket{2}$  &~~~~~$\cdots$&~~~~~ $\ket{m}$ &~$:=\ket{\zeta_{\pm}^{1,1}}$\\ 
$\ket{\gamma_{\pm}^1}$ &~~~~~$\ket{p}$  &~~~~~$\ket{2}$  &~~~~~$\cdots$&~~~~~ $\ket{m}$ &~$:=\ket{\zeta_{\pm}^{1,2}}$\\
$\ket{p}$ &~~~~~$\ket{\epsilon_{\pm}}$  &~~~~~$\ket{2}$  &~~~~~$\cdots$&~~~~~ $\ket{m}$ &~$:=\ket{\zeta_{\pm}^{1,3}}$\\  
$\ket{2}$ &~~~~~$\ket{1}$  &~~~~~$\ket{\gamma_{\pm}^2}$  &~~~~~$\cdots$&~~~~~ $\ket{m}$ &~$:=\ket{\zeta_{\pm}^{2,1}}$\\
$\ket{\gamma_{\pm}^2}$ &~~~~~$\ket{1}$  &~~~~~$\ket{p}$  &~~~~~$\cdots$&~~~~~ $\ket{m}$ &~$:=\ket{\zeta_{\pm}^{2,2}}$\\
$\ket{p}$ &~~~~~$\ket{1}$  &~~~~~$\ket{\epsilon_{\pm}}$  &~~~~~$\cdots$&~~~~~ $\ket{m}$ &~$:=\ket{\zeta_{\pm}^{2,3}}$\\
$\vdots$&~~~~~$\vdots$  &~~~~~$\vdots$  &~~~~~$\vdots$  &~~~~~$\vdots$  &~~~~~$\vdots$\\
$\ket{m}$ &~~~~~$\ket{1}$  &~~~~~$\ket{2}$  &~~~~~$\cdots$&~~~~~ $\ket{\gamma_{\pm}^m}$ &~$:=\ket{\zeta_{\pm}^{m,1}}$\\
$\ket{\gamma_{\pm}^m}$ &~~~~~$\ket{1}$  &~~~~~$\ket{2}$  &~~~~~$\cdots$&~~~~~ $\ket{p}$ &~$:=\ket{\zeta_{\pm}^{m,2}}$\\
$\ket{p}$ &~~~~~$\ket{1}$  &~~~~~$\ket{2}$  &~~~~~$\cdots$&~~~~~ $\ket{\epsilon_{\pm}}$ &~$:=\ket{\zeta_{\pm}^{m,3}}$\\
\hline 
\end{tabular}
\caption{Set of states $\mathbb{G}[(m+2)\otimes3^{\otimes m}]$. Here $\ket{\epsilon_{\pm}}:=\frac{1}{\sqrt{2}}(\ket{p}\pm\ket{q})$ and  $\ket{\gamma_{\pm}^i}:=\frac{1}{\sqrt{2}}(\ket{q}\pm\ket{i})$; with $i,j\in\{p,q,1,\cdots,m\}$ and $\braket{i|j}=\delta_{ij}$.}\label{tab1}
\end{table} 
\begin{proposition}\label{prop2}
Consider the set of states $\mathbb{G}[(m+2)\otimes3^{\otimes m}]\equiv\left\lbrace\ket{\zeta_{\pm}^0},\ket{\zeta_{\pm}^{i,1}},\ket{\zeta_{\pm}^{i,2}},\ket{\zeta_{\pm}^{i,3}}\right\rbrace_{i=1}^m$ given in the Table \ref{tab1}. This set is a GNPS in $\mathbb{C}^{m+2}\otimes\left(\mathbb{C}^3\right)^{\otimes m}$. Here, Alice posses the subsystem in $\mathbb{C}^{m+2}$ and each Bob has a subsystem with qutrit Hilbert space. 
\end{proposition}
\begin{proof}
Consider the subset of states $\$^i\equiv\left\lbrace \ket{\zeta_{\pm}^0},\ket{\zeta^{i,1}_{\pm}},\ket{\zeta^{i,2}_{\pm}},\ket{\zeta^{i,3}_{\pm}}\right\rbrace\subset\mathbb{G}[(m+2)\otimes3^{\otimes m}]$, for $i\in\{1,\cdots,m\}$. The set $\$^i$ has similar structure as of the set $\mathbb{S}_{Ben}$ between Alice and $i^{th}$ Bob while other Bobs have fixed states tagged with this set (see Fig. \ref{fig2}). This, along with the Observation \ref{obs1} and Observation \ref{obs2}, assure that the set $\mathbb{G}[(m+2)\otimes3^{\otimes m}]$ cannot be locally distinguished in any bipartition.
\end{proof}
\begin{figure}[t!]
\begin{center}
\includegraphics[scale=0.5]{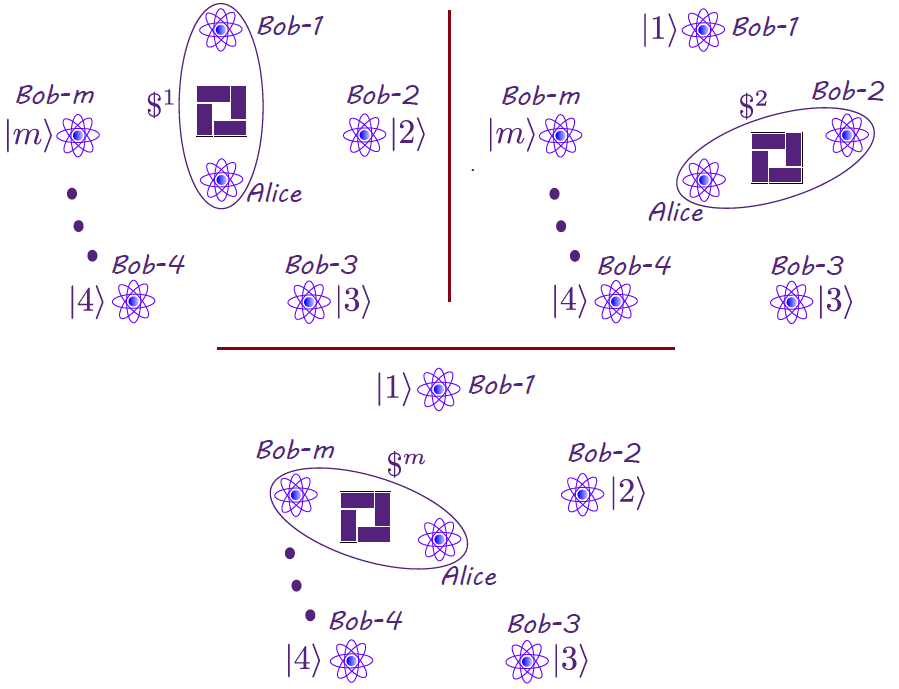}
\caption{Multiparty GNPS $\mathbb{G}[(m+2)\otimes3^{\otimes m}]$. Alice \& $i^{th}$ Bob share the set of states $\$^i\equiv\left\lbrace \ket{\zeta_{\pm}^0},\ket{\zeta^{i,1}_{\pm}},\ket{\zeta^{i,2}_{\pm}},\ket{\zeta^{i,3}_{\pm}}\right\rbrace$ while $j$th Bob's state be $\ket{j}$, $j\ne i$. Clearly, $\mathbb{G}[(m+2)\otimes3^{\otimes m}]\equiv\bigcup_i\left\lbrace\bigotimes_{j\ne i}\ket{j}_{B_j}\otimes\$^i_{AB_i}\right\rbrace$.}\label{fig2}
\end{center}
\end{figure}
Our next result addresses entanglement assisted discrimination of the set $\mathbb{G}[(m+2)\otimes3^{\otimes m}]$.   
\begin{theorem}\label{theo2}
The set of states $\mathbb{G}[(m+2)\otimes3^{\otimes m}]$ can be perfectly discriminated locally given the genuine resource state $\ket{g_{m+1}}_{ab_1\cdots b_m}:=\frac{1}{\sqrt{2}}\left(\ket{0^{\otimes m+1}}+\ket{1^{\otimes m+1}}\right)_{ab_1\cdots b_m}$. 
\end{theorem}  
This proof follows straightforwardly by generalizing the discrimination strategy discussed in the proof of Theorem \ref{theo1}. For completeness we provide the proof in Appendix. It is important to note that, here also the discriminating resource lives in the minimal Hilbert space dimension, {\it i.e.} in $(\mathbb{C}^2)^{\otimes m+1}$, and hence the genuine nonlocality of the GNPS $\mathbb{G}[(m+2)\otimes3^{\otimes m}]$ can be considered minimal. 

We will now discuss another important implication of our construction. In multipartite scenario, one of the most pertinent problems is the resource comparison among different entangled states. One possible way is to check the possible  inter-conversion between two states under LOCC. However, there exist states that are not comparable in this sense. For instance, considered the states $\ket{g_3}\in(\mathbb{C}^2)^{\otimes3}$ and $\ket{\psi}\equiv\ket{\chi}\otimes\ket{\eta}$, with $\ket{\chi}\in\mathbb{C}^{d_1}\otimes\mathbb{C}^{d_2}$ having Schmidt rank greater than two and $\ket{\eta}\in\mathbb{C}^{d_3}$. Clearly, $\ket{g_3}$ being a genuine entangled state cannot be obtained from the biseparable state $\ket{\psi}$ under LOCC. On the other hand, neither a deterministic \cite{Nielsen99} nor a probabilistic \cite{Vidal99} transformation form the state $\ket{g_3}$ to the state $\ket{\psi}$ is possible even if entanglement of $\ket{\psi}$ is strictly less than unity \cite{Self1}. At this point some task ($\tau$) based ordering relation $\succ_\tau$ might be of interest. We will say the state $\rho$ is better than the state $\sigma$ in performing the task $\tau$, if the task can be perfectly done with the state $\rho$ as a resource but not with $\sigma$ and hence it induces an operational ordering between the states represented as $\rho\succ_\tau\sigma$. In that sense, our construction suggests the following ordering relation. 
\begin{corollary}\label{coro1}
The task ($\tau_m$) of entanglement assisted discrimination of the set $\mathbb{G}[(m+2)\otimes3^{\otimes m}]$ induces the ordering relation $\ket{g_{m+1}}\succ_{\tau_m}\rho:=\sum p_i\chi^i\otimes\eta^i$, where $\forall~i,~ \chi^i\in\mathcal{D}(\otimes_{j=1}^m\mathcal{H}_j)$ and $\eta^i\in\mathcal{D}(\mathcal{H})$ with $\mathcal{H}_j$'s and $\mathcal{H}$ having arbitrary dimension; and $p_i\ge0,~\sum p_i=1$. 
\end{corollary}
Here, $\mathcal{D}(X)$ denotes the set of density operator acting on $X$. Note that the state $\rho$ can have at-most $m$ partite entanglement which makes the proof of Corollary \ref{coro1} immediate. Consider now the tripartite resource state $\ket{\psi_3}:=\ket{g_3}_{abc}^{\otimes2}$, {\it i.e.} two copies of three-qubit GHZ state shared among three parties; and the state $\ket{\phi_3}:=\ket{\phi^+}_{ab}\otimes\ket{\phi^+}_{bc}\otimes\ket{\phi^+}_{ca}$, {\it i.e.} three copies of two-qubit maximally entangled state $\ket{\phi^+}:=\frac{1}{\sqrt{2}}(\ket{00}+\ket{11})$ symmetrically shared among three parties. Both $\ket{\psi_3}$ and $\ket{\phi_3}$ contain tripartite genuine entanglement and both the states have same single party marginal. Moreover, these two resources are incomparable under LOCC \cite{Bennett00}; in-fact it is not possible to convert $2N$ three-party GHZ states into $3N$ singlets even in an asymptotic sense \cite{Linden05}. At this point, consider the task $\tau^\star$ of distinguishing the ordered pair of states $(\ket{\zeta_i},\ket{\zeta_j})$ chosen randomly from the Cartesian product set $\mathbb{G}[4\otimes3^{\otimes2}]\times\mathbb{G}[4\otimes3^{\otimes2}]$. Our next result, brings a bona fide ordering between the locally incomparable genuine resource states $\ket{\psi_3}$ and $\ket{\phi_3}$. 
\begin{corollary}\label{coro2}
The tripartite product state discrimination problem $\tau^\star$ induces the ordering relation $\ket{\psi_3}\succ_{\tau^\star}\ket{\phi_3}$. 
\end{corollary}
\begin{proof}
The task $\tau^\star$ considers discrimination of the ordered tuple $\left( \ket{\zeta_i},\ket{\zeta_j}\right)$ chosen randomly from $\mathbb{G}[4\otimes3^{\otimes2}]\times\mathbb{G}[4\otimes3^{\otimes2}]$. Clearly, the task cannot be done under LOCC. An additional resource $\ket{\phi_3}$ also fails to achieve the desired objective perfectly. The set $\mathbb{G}[4\otimes3^{\otimes2}]$ being a GNPS necessitates consumption of at least two of the three symmetrically distributed EPR states for perfect discrimination of the first element of the ordered pair $\left( \ket{\zeta_i},\ket{\zeta_j}\right)$. Since, identification of the first element does not provide any information regarding the second, therefore it cannot be perfectly discriminated using the remaining one EPR state. However, given the resource $\ket{\psi_3}$, two copies of three qubit GHZ, the players can use the first and second copy respectively to perfectly identify $\ket{\zeta_i}$ and $\ket{\zeta_j}$. This can be done by following the protocol discussed in Theorem \ref{theo1}. This completes the proof.  
\end{proof}
Furthermore, following the construction of bipartite unextendible product bases of Ref. \cite{Halder19(1)}, the present construction can be further generalized for higher dimensional Hilbert spaces. For the explicit construction we refer to the Appendix. There we construct a GNPS in $\mathbb{C}^6\otimes(\mathbb{C}^5)^{\otimes2}$. It might be interesting to see whether a resource efficient discrimination protocol is possible for this set.    

\section{Discussions}
We have constructed genuinely nonlocal product bases for arbitrary many number of parties. We then argued that the strength of nonlocality of those sets can be considered minimal as they require entangled resource of minimal dimension for their perfect discrimination. The constructions also lead to fully separable measurements whose implementation require either all the parties to come together or they need to share some multipartite resource that contains entanglement across all possible bipartite cuts. 

Our study also motivates some interesting questions for further research. While we have considered entangled resource of minimal dimension, the question remain open which particular entangled state in this minimum dimensional Hilbert space turns out to be the optimal resource. In this respect, constructing a tripartite GNPS that can be perfectly distinguished with the resource of $3$-qubit W state might be of particular interest. 

\section*{Acknowledgments} 
SR acknowledges partial support by the Foundation for Polish Science (IRAP project, ICTQT, contract no. MAB/2018/5, co-financed by EU within Smart Growth Operational Programme). MB acknowledges support through INSPIRE Faculty fellowship by the Department of Science and Technology, Government of India. 

\onecolumngrid 
\appendix
\onecolumngrid 
\section{Proof of Theorem \ref{theo2}}
\begin{proof}
In {\bf Step-1} Alice performs the measurement $\mathcal{M}\equiv\{M,\mathbb{I}-M\}$, where
\begin{eqnarray*}
M:=\mathbb{P}\left[\ket{p}_A;\ket{0}_{a}\right]+\mathbb{P}\left[\left( \ket{q},\ket{1},\cdots,\ket{m}\right)_A;\ket{1}_{a}\right].
\end{eqnarray*}
The evolved states are given by,
\begin{align}\nonumber
\rotatebox[origin=c]{0}{}
\left\{\!\begin{aligned}
\left\lbrace\ket{\zeta^{1,3}_{\pm}},\cdots,\ket{\zeta^{m,3}_{\pm}} \right\rbrace\otimes\ket{0^{\otimes m+1}}_{ab_1\cdots b_m},\\
\left\lbrace \ket{\zeta^{1,1}_{\pm}},\ket{\zeta^{1,2}_{\pm}},\cdots,\ket{\zeta^{m,1}_{\pm}},\ket{\zeta^{m,2}_{\pm}}\right\rbrace\otimes\ket{1^{\otimes m+1}}_{ab_1\cdots b_m},\\
\ket{\zeta_{\pm}^0}\Rightarrow\ket{\tilde{\zeta}^0_{\pm}}~~~~~~~~~~~~~~~~~~ 
\end{aligned}\right\},	
\end{align}
where, 
\begin{eqnarray}
\ket{\tilde{\zeta}^0_{\pm}}&=&\left(\ket{p}_A\ket{0^{\otimes m+1}}_{ab_1\cdots b_m}\pm\ket{q}_A\ket{1^{\otimes m+1}}_{ab_1\cdots b_m}\right)\otimes\ket{1}_{B_1}\cdots\ket{m}_{B_m}.
\end{eqnarray}
{\bf Step-2:} $i^{th}$ Bob performs the similar measurement as in Theorem $1$. If $K^i_3$ clicks the state is one of $\left\lbrace \ket{\zeta^{i,3}_{\pm}} \right\rbrace$, if $K^i_2$ clicks the state is one of $\left\lbrace \ket{\zeta^{i,2}_{\pm}}\right\rbrace$, if all $K^i_1$'s click the state is one of $\left\lbrace \ket{\zeta_{\pm}^0},\ket{\zeta^{i,1}_{\pm}}\right\rbrace_{i=1}^m $. Alice then performs the measurement,
\begin{align}\nonumber
\rotatebox[origin=c]{0}{$\mathcal{M}^\prime\equiv$}
\left\{\!\begin{aligned}
M^\prime_1:=\mathbb{P}\left[\ket{1}_{A};\mathbb{I}_{a}\right],\cdots
M^\prime_m:=\mathbb{P}\left[\ket{m}_{A};\mathbb{I}_{a}\right],\\
M^\prime_0:=\mathbb{I}-\left(M^\prime_1+\cdots+M^\prime_2\right) .~~~~~~~~~~~~
\end{aligned}\right\}.	
\end{align}
If $M^\prime_i$ clicks the state is one of $\left\lbrace \ket{\zeta^{i,1}_{\pm}}\right\rbrace $, else it is one of  $\left\lbrace \ket{\zeta_{\pm}^0}\right\rbrace$.  Now the result in \cite{Walgate00} assures local distinguishability between any two orthogonal states. 
\end{proof}

\section{Construction of GNPS in higher dimensional Hilbert spaces}
The NPS $\mathbb{S}_{Ben}\subset\mathbb{C}^3\otimes\mathbb{C}^3$ can be expressed in the following generic form,
\begin{equation}\label{b2qg}
\mathbb{S}_{Ben}\equiv\mathbb{S}_{Ben}[3\otimes3]\equiv\left\lbrace \topt{a}{s_\pm},\topt{s_\pm}{c},\topt{c}{t_\pm},\topt{t_\pm}{a}\right\rbrace,
\end{equation}
where $\{\ket{a},\ket{b},\ket{c}\}$ are pairwise orthonormal states and $\ket{s_\pm}:=\frac{1}{\sqrt{2}}\left(\ket{a}\pm\ket{b}\right)$ and $\ket{t_\pm}:=\frac{1}{\sqrt{2}}\left(\ket{b}\pm\ket{c}\right)$. A generalization of $\mathbb{S}_{Ben}[3\otimes3]$ in $\mathbb{C}^5\otimes\mathbb{C}^5$ is given by,
\begin{align}\label{b4qg}
\rotatebox[origin=c]{0}{$\mathbb{S}_{Ben}[5\otimes5]\equiv$}
\left\{\!\begin{aligned}
\ket{\Gamma^1_\pm}:=\topt{a}{s_\pm},~~\ket{\Gamma^2_\pm}:=\topt{s_\pm}{c},~~\\
\ket{\Gamma^3_\pm}:=\topt{c}{t_\pm},~~\ket{\Gamma^4_\pm}:=\topt{t_\pm}{a},~~\\
\ket{\Gamma^5_{ijk}}:=\topt{d}{u_{ijk}},\ket{\Gamma^6_{ijk}}:=\topt{u_{ijk}}{e},\\
\ket{\Gamma^7_{ijk}}:=\topt{e}{v_{ijk}},\ket{\Gamma^8_{ijk}}:=\topt{v_{ijk}}{d}~~
\end{aligned}\right\},	
\end{align} 
where $\{\ket{a},\ket{b},\ket{c},\ket{d},\ket{e}\}$ is an orthonormal basis of $\mathbb{C}^5$ and $\ket{u_{ijk}}\in\mathcal{S}_{abcd}$ and $\ket{v_{ijk}}\in\mathcal{S}_{abce}$, with
\begin{align}\nonumber
\rotatebox[origin=c]{0}{$\bar{\mathcal{S}}_{\alpha\beta\delta\gamma}\equiv$}
\left\{\!\begin{aligned}
\ket{\alpha}+(-1)^i\ket{\beta}+(-1)^j\ket{\delta}+(-1)^k\ket{\gamma},\\
\mbox{with}~i,j,k\in\{0,1\}~and~i\oplus_2j\oplus_2k=0
\end{aligned}\right\}.	
\end{align}
$\bar{\mathcal{S}}_{\alpha\beta\delta\gamma}$ contains the unnormalized states of $\mathcal{S}_{\alpha\beta\delta\gamma}$. The NPS $\mathbb{S}_{Ben}[5\otimes5]$ has a \textit{layered} tile structure (see Fig. \ref{tile5}). This has been recently studied to understand the intricate geometrical structure of the set of bipartite states having positive partial transpose, {\it i.e.}, the {\it Peres set} \cite{Halder19(1)}. Furthermore, from Ref.\cite{Zhang20} it is evident that the set (\ref{b4qg}) can be locally distinguished if a $2$-qutrit maximally entangled state is shared as resource. Note that the protocol in \cite{Zhang20} is resource efficient compared to the teleportation based protocol as the later requires a maximally entangled state of $\mathbb{C}^5\otimes\mathbb{C}^5$

Consider now the following set of states in $\mathbb{C}^6_A\otimes\mathbb{C}^5_{B_1}\otimes\mathbb{C}^5_{B_2}$,
\begin{align}\label{c11}
\rotatebox[origin=c]{0}{$\mathbb{G}[6\otimes5^{\otimes2}]\equiv$}
\left\{\!\begin{aligned}
\ket{\Omega^1_\pm}&:=\tri{1}{\alpha_\pm}{4},~\ket{\Omega^2_\pm}:=\tri{\alpha_\pm}{3}{4},\\
\ket{\Omega^3_\pm}&:=\tri{3}{\beta_\pm}{4},~\ket{\Omega^4_\pm}:=\tri{\beta_\pm}{1}{4},\\
\ket{\Omega^5_\pm}&:=\tri{4}{3}{\gamma_\pm},~\ket{\Omega^6_\pm}:=\tri{\gamma_\pm}{3}{1},\\
\ket{\Omega^7_\pm}&:=\tri{1}{3}{\alpha_\pm},~\ket{\Omega^8_{ijk}}:=\tri{0}{\Psi_{ijk}}{3},~\\
\ket{\Omega^9_{ijk}}&:=\tri{\Psi_{ijk}}{4}{3},\ket{\Omega^{10}_{ijk}}:=\tri{4}{\Phi_{ijk}}{3},\\
\ket{\Omega^{11}_{ijk}}&:=\tri{\Phi_{ijk}}{0}{3},\ket{\Omega^{12}_{ijk}}:=\tri{5}{0}{\Phi_{ijk}},\\
\ket{\Omega^{13}_{ijk}}&:=\tri{\varUpsilon_{ijk}}{0}{5},\ket{\Omega^{14}_{ijk}}:=\tri{3}{0}{\varUpsilon_{ijk}}
\end{aligned}\right\},	
\end{align}
where $$\ket{\alpha_\pm}:=\frac{1}{\sqrt{2}}\ket{1\pm2}, \ket{\beta_\pm}:=\frac{1}{\sqrt{2}}\ket{2\pm3}, \ket{\gamma_\pm}:=\frac{1}{\sqrt{2}}\ket{2\pm4}$$ $$\ket{\Psi_{ijk}}\in\mathcal{S}_{0123},~\ket{\Phi_{ijk}}\in\mathcal{S}_{1234},~ \ket{\varUpsilon_{ijk}}\in\mathcal{S}_{1245}.$$
\begin{figure}[t!]
\centering
\includegraphics[scale=0.20]{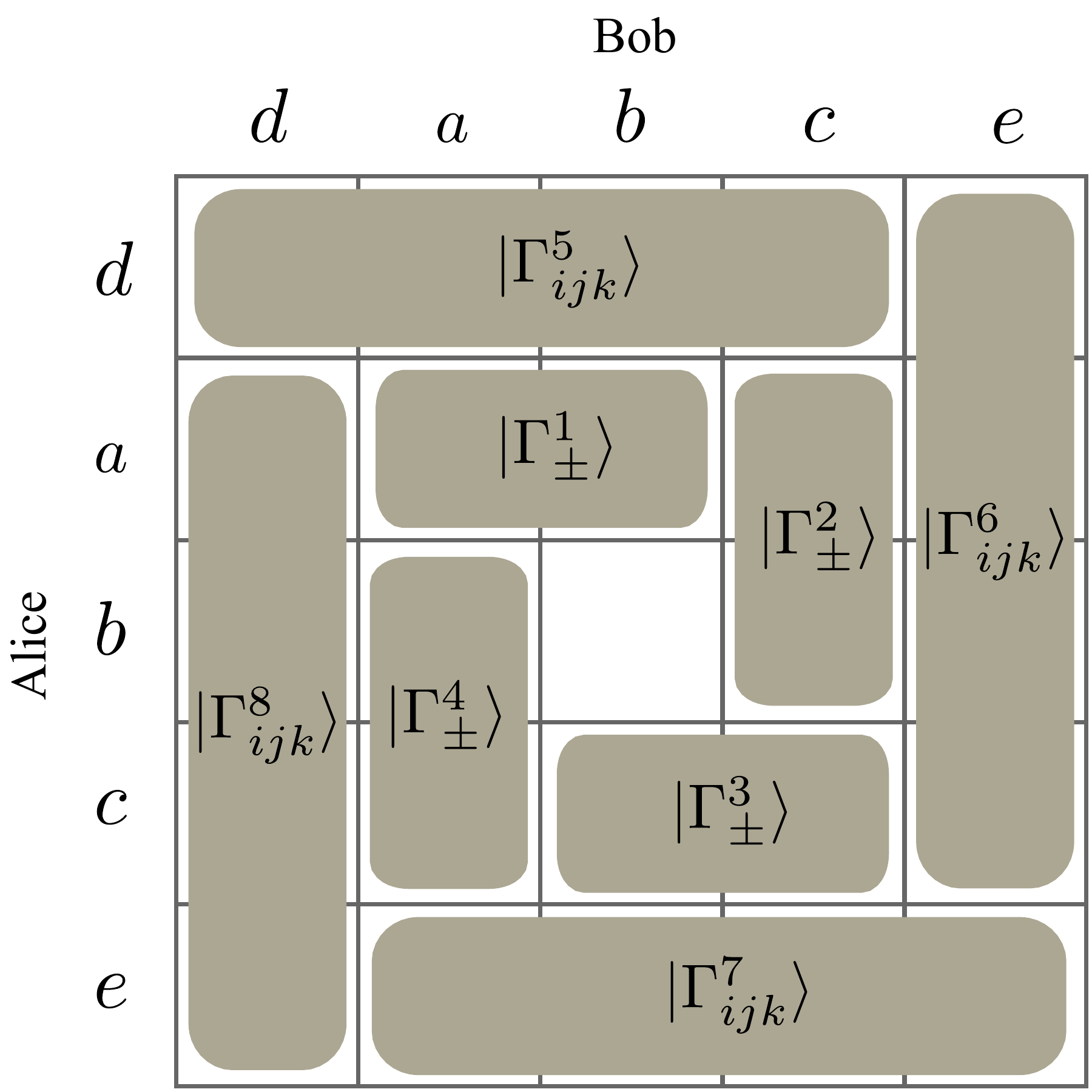}
\caption{Tile structure of the set $\mathbb{S}_{Ben}[5\otimes5]$. Cardinality of the set is $24$. Each inner layered tile contains $2$ mutually orthonormal states, while each outer layered tile contains $4$ mutually orthonormal states. Orthogonality among the states from different tiles be evident from the structure.}\label{tile5}
\end{figure}

\begin{figure}[t!]
	\centering
	\includegraphics[scale=0.2]{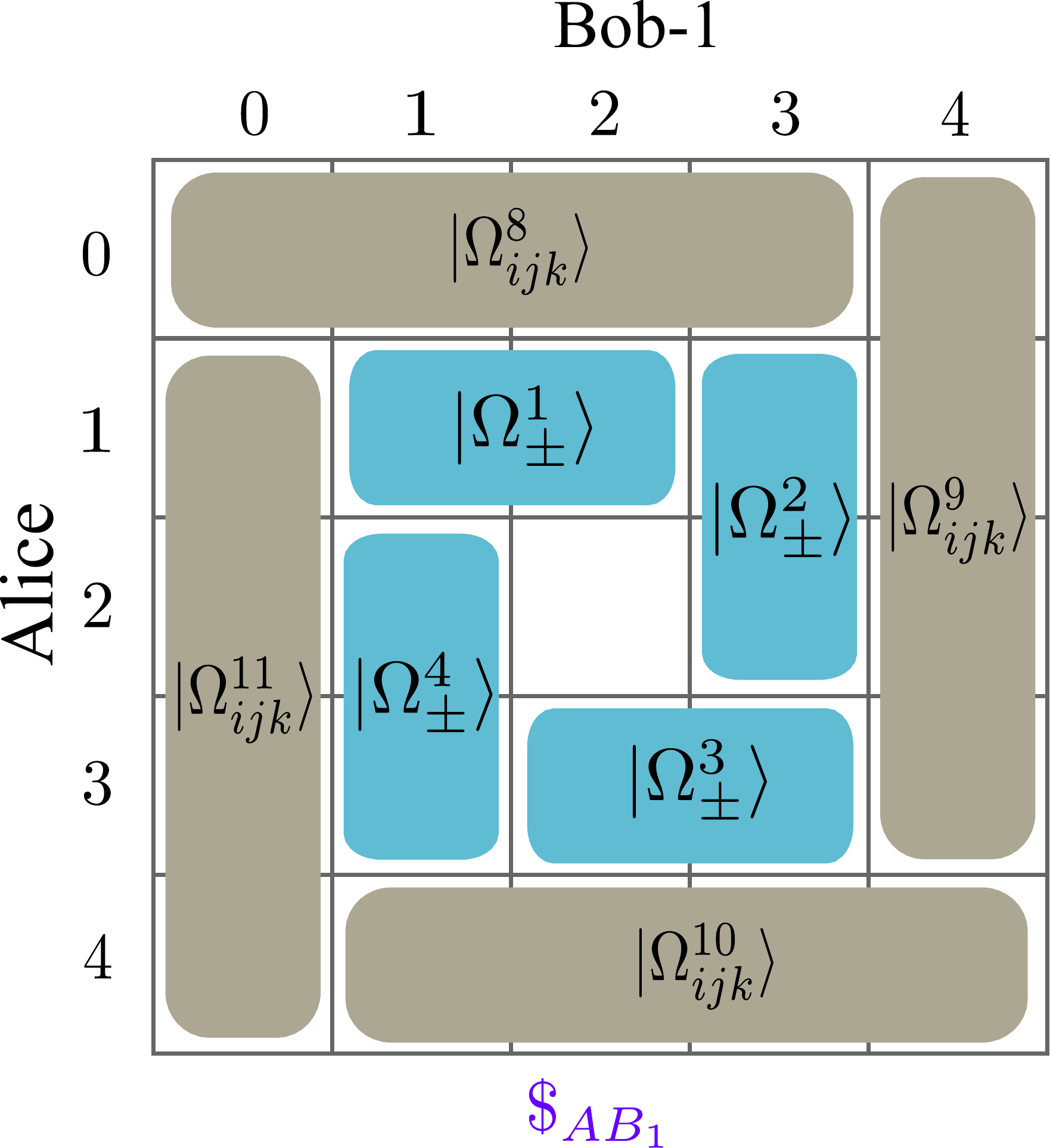}~~~~~~~~~~~~~\includegraphics[scale=0.2]{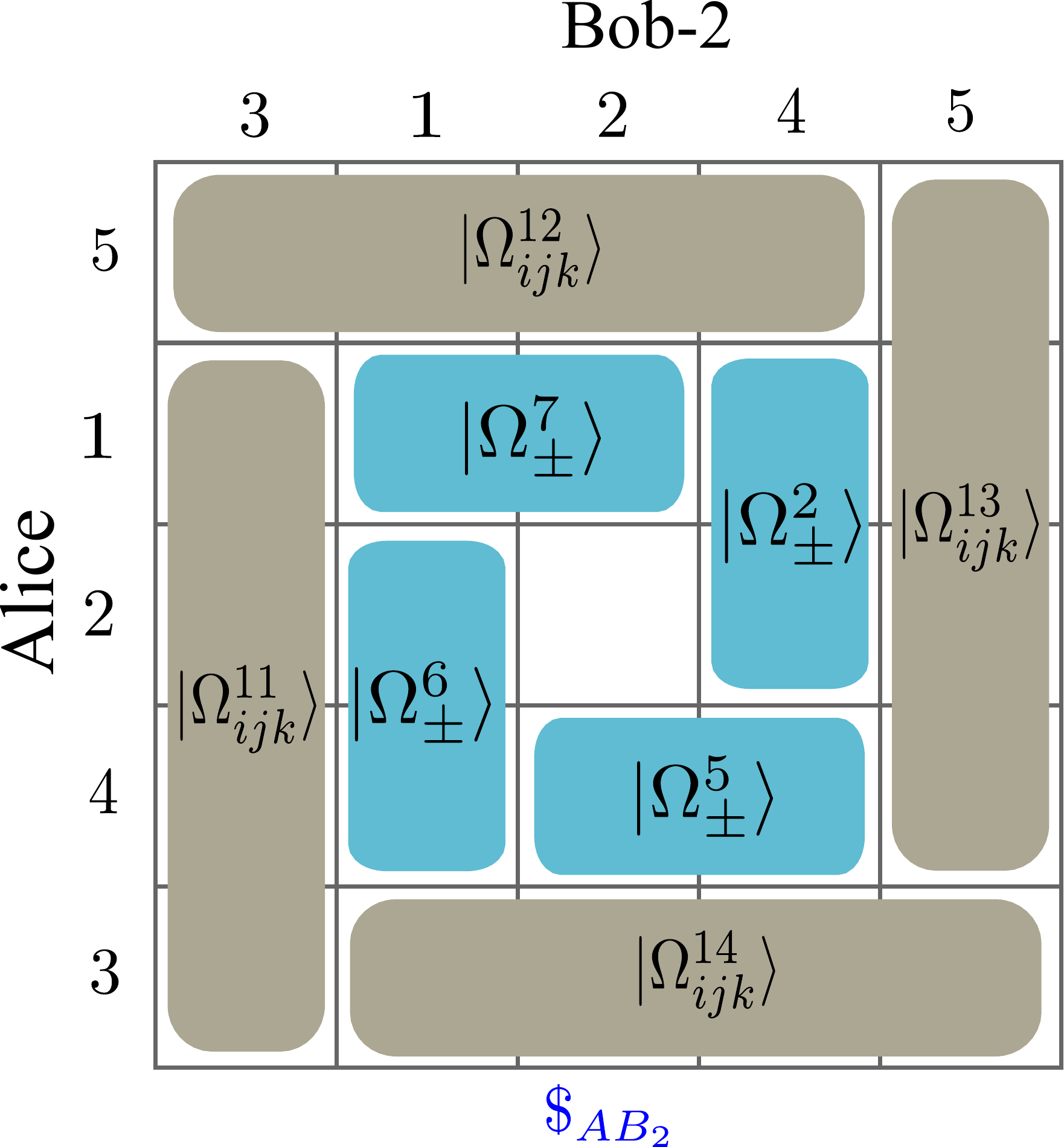}
	\caption{[Color on-line] Tile structure of the set $\$_{AB_1}$ (left) and $\$_{AB_2}$ (right). With all the states in outer layered (grey) tiles of $\$_{AB_1}$ Bob-2's state is $\ket{4}_{B_2}$, while for inner layer (blue) his state is $\ket{3}_{B_2}$. In $\$_{AB_2}$, Bob-1's state tagged with outer layer is $\ket{0}_{B_1}$ and for inner layer it is $\ket{3}_{B_1}$.}\label{tile51}
\end{figure}

Before proceeding further, let us first analyze the structure of the set $\mathbb{G}[6\otimes5^{\otimes2}]$. The subset $\$_{AB_1}\equiv\left\lbrace \ket{\Omega^1_\pm},\ket{\Omega^2_\pm},\ket{\Omega^3_\pm},\ket{\Omega^4_\pm},\ket{\Omega^8_{ijk}},\ket{\Omega^9_{ijk}},\ket{\Omega^{10}_{ijk}},\ket{\Omega^{11}_{ijk}}\right\rbrace $ has a kind of analogous structure as of (\ref{b4qg}) between Alice and Bob-1 (see Fig.\ref{tile51}). Please note here an important point: Bob-2 has the state $\ket{4}_{B_2}$ tagged with $\{\ket{\Omega^1_\pm},\ket{\Omega^2_\pm},\ket{\Omega^3_\pm},\ket{\Omega^4_\pm}\}$, while with $\{\ket{\Omega^8_{ijk}},\ket{\Omega^9_{ijk}},\ket{\Omega^{10}_{ijk}},\ket{\Omega^{11}_{ijk}}\}$ Bob-2's state $\ket{3}_{B_2}$ is tagged. Similarly, $\$_{AB_2}\equiv\left\lbrace\ket{\Omega^2_\pm},\ket{\Omega^5_\pm},\ket{\Omega^6_\pm},\ket{\Omega^7_\pm},\ket{\Omega^{11}_{ijk}},\ket{\Omega^{12}_{ijk}},\ket{\Omega^{13}_{ijk}},\ket{\Omega^{14}_{ijk}}\right\rbrace $ has a kind of similar structure as of (\ref{b4qg}) between Alice and Bob-2 with Bob-1 having the tagged state $\ket{3}_{B_1}$ with $\{\ket{\Omega^2_\pm},\ket{\Omega^5_\pm},\ket{\Omega^6_\pm},\ket{\Omega^7_\pm}\}$ and having the tagged state $\ket{0}_{B_1}$ with $\{\ket{\Omega^{11}_{ijk}},\ket{\Omega^{12}_{ijk}},\ket{\Omega^{13}_{ijk}},\ket{\Omega^{14}_{ijk}}\}$. This structure, along with the Observations $1$ and $2$ discussed in manuscript, leads us to the following proposition.\\\\
{\bf Proposition 3.} {\it The set of states $\mathbb{G}[6\otimes5^{\otimes2}]$ is a GNPS in $\mathbb{C}^6\otimes\mathbb{C}^5\otimes\mathbb{C}^5$.}

\twocolumngrid


\begin{thebibliography}{10}
\bibitem{Self1} Although $\ket{g_3}$ and $\ket{\psi}$ are LOCC incomparable under one-copy manipulation (which is the topic we are focusing here), given many copies of $\ket{g_3}$ one can, however, obtain $\ket{\psi}$ under LOCC, but the converse is a strict impossibility.
\end{thebibliography}
\end{document}